\batchmode
\documentclass[12pt]{article}    
\usepackage[english]{babel}
\usepackage{hyperref}
\usepackage{graphics}
\usepackage{epsfig}
\usepackage{natbib}   
\usepackage{graphicx} 
\usepackage{amssymb}  
\usepackage{multirow} 
\usepackage{epstopdf}
\usepackage{bm}
\usepackage{amsmath}
\usepackage{amsfonts}
\usepackage{amssymb}
\usepackage{amssymb, longtable}
\usepackage{pdflscape}
\usepackage{color}
\begin{document}
\newtheorem{definition}{Definition}
\newtheorem{theorem}{Theorem}
\newtheorem{example}{Example}
\newtheorem{corollary}{Corollary}
\newtheorem{lemma}{Lemma}
\newtheorem{proposition}{Proposition}
\newenvironment{proof}{{\bf Proof:\ \ }}{\qed}
\newcommand{\qed}{\rule{0.5em}{1.5ex}}
\newcommand{\bfg}[1]{\mbox{\boldmath $#1$\unboldmath}}

\begin{center}

\section*{Aggregation of Dependent Risks in Mixtures of Exponential Distributions and Extensions}

\vskip 0.2in {\sc \bf Jos\'e Mar\'{\i}a Sarabia$^a$\footnote{Corresponding author. E-mail addresses: sarabiaj@unican.es (J.M. Sarabia); emilio.gomez-deniz@ulpgc.es (E. G\'omez-D\'eniz); faustino.prieto@unican.es (F. Prieto); vanesa.jorda@unican.es (V. Jord\'a).}, Emilio G\'omez-D\'eniz$^b$, \\ Faustino Prieto$^a$, Vanesa Jord\'a$^a$}

\small{

\vskip 0.15in

$^a$Department of Economics, University of Cantabria, Avda de los Castros s/n, 39005-Santander, Spain
\vskip 0.1in
$^b$Department of Quantitative Methods in Economics and TiDES Institute, University of Las Palmas de Gran Canaria, 35017-Las Palmas de G.C., Spain
}

\end{center}

\vskip 0.1in

\begin{abstract}\noindent
The distribution of the sum of dependent risks is a crucial aspect in actuarial sciences, risk management and in many branches of applied probability. In this paper, we obtain analytic expressions for the probability density function (pdf) and the cumulative distribution function (cdf) of aggregated risks, modeled according to a mixture of exponential distributions. We  first review the properties of the multivariate mixture of exponential distributions, to then obtain the analytical formulation for the pdf and the cdf for the aggregated distribution. We study in detail some specific families with Pareto (Sarabia et al, 2016), Gamma, Weibull and inverse Gaussian mixture of exponentials (Whitmore and Lee, 1991) claims. We also discuss briefly the computation of risk measures, formulas for the ruin probability (Albrecher et al., 2011) and the collective risk model. An extension of the basic model based on mixtures of gamma distributions is proposed, which is one of the suggested directions for future research.
\end{abstract}

\vskip 0.15in

\noindent\textit{Key Words}: Aggregation, dependent random variables, Laplace transform, collective risk model, partial Bell polynomials

\newpage
\section{Introduction}

The distribution of the sum of risks is a key aspect in actuarial sciences, risk management and in many branches of applied probability. Albeit it can be easily obtained for independent risks, this assumption is in most cases too restrictive, thus being crucial to specify more general models that allow for dependence between different risks. In the recent statistical and actuarial literature, several results about risk aggregation under dependence have been obtained, which deploy different copula structures (see, e.g., Arbenz et al. (2012), Coqueret (2014), Gijbels and Herrmann (2014)).

Cossette et al. (2013) consider risk aggregation and capital allocation problems for a portfolio of dependent risks, modeling the multivariate distribution with the Farlie-Gumbel-Morgenstern (FGM) copula and mixed Erlang distribution marginals. Hashorva and Ratovomirija (2015) considers and extension of previous model introducing the Sarmanov distribution to model the dependence structure, also demonstrating that the aggregated risk belongs to the class of Erlang mixtures. Also using the Sarmanov's distribution to define the dependence structure, Vernic (2016) presents some formulas for the density of the sum of random variables, with a particular focus on exponentially distributed marginals. The multivariate Pareto distributions seems to be an outstanding candidate to model dependent risk. Sarabia et al. (2016) have studied aggregation in multivariate dependent Pareto distributions, also providing closed formulas for the individual risk model, and for the collective risk model assuming Poisson, negative binomial and logarithmic as primary distributions. B{\o}lviken and Guillen (2017) have also considered the Pareto copula, using in this case log-normal marginals, to study risk aggregation, improving the accuracy of the model by updating the skweness recursively. A flexible approach has been recently proposed by C\^ot\'e and Genest (2015), which consist of a consists of a tree structure of bivariate copulas, which are assumed to be conditionally independence.

In this paper, we propose to model the aggregated risk  by a multivariate mixture of exponential distributions. The model proposed by Lindley and Singpurwalla (1986) corresponds to gamma mixtures of exponential distributions. Extensions of this model were provided by Nayak (1987) and Roy and Mukherjee (1988) among others. This model has been widely used in actuarial science, reliability studies and quantitative risk analysis. For instance, it was considered by Albrecher et al. (2011), providing explicit formulas for the ruin probability. However, the question of the aggregation of risks in this model remains open. A particular case was considered in Sarabia et al. (2016), where the mixing distribution is of gamma type. The general multivariate mixture of exponential distributions can be specified for any positive mixing distribution described in terms of Laplace transform. Using this transform, we obtain analytic expressions for the probability density function (pdf) and the cumulative distribution function (cdf) of the aggregated risks.

The rest of the paper is organized as follows. In Section \ref{section2}, we review the properties of the multivariate mixture of exponentials including a characterization theorem (in terms of the copula generator and the marginal distributions), dependence conditions (total positivity of order two in pairs and associated random variables), dependence measures, moments, copula (which belong to the Archimedean family) and other relevant features. In Section \ref{section3}, we obtain the analytical formulation for the pdf and the cdf of the aggregated distribution. Moroever, we include expressions for the survival function, the moments and the value at risk (VaR). Different models are studied in Section \ref{section4}, with a primary focus on claims of Pareto type (Sarabia et al., 2016), Gamma, Weibull with shape parameter 1/2, general Weibull and inverse Gaussian mixture of exponentials (Whitmore and Lee, 1991). In particular, the models with Pareto and Weibull claims have Clayton and Gumbel copulas, respectively. As regards the other multivariate models, their dependence structure is characterized by new families of copulas, which are obtained. For all these models, we obtain specific expressions for the aggregated distribution, and we study some of their main properties. We also discuss briefly the computation of several risk measures including VaR, tail value at risk and other tail measures, formulas for the ruin probability (Albrecher et al., 2011) and the collective risk model. In Section \ref{section5}, we study an extension of the basic multivariate model, given by a multivariate mixture of classical gamma distributions. A specific model is considered,  based on mixtures of gamma and gamma product-ratio claims (Sibuya, 1979). Finally, in Section \ref{section6}, we include some conclusions and ideas about future research.

\section{The model}\label{section2}

Let $\Theta$ be a positive random variable with cdf $F_\Theta(\cdot)$ and  Laplace-Stieltjes transform (hereinafter referred to as Laplace transform) $L_{\Theta}(\cdot)$, that is $L_\Theta(s)=E[\exp(-s\Theta)]=\int_0^\infty e^{-sz}dF_\Theta(z)$. A distribution with support on $(0,\infty)$ is identified by its Laplace transform. We consider the classical compound Poisson risk with exponential claims sizes (Asmussen and Albrecher (2010), Klugman et al. (2008)). Hence, we have a random vector $(X_1,\dots,X_n)$, which comprises independent exponentially distributed random variables conditionally on $\Theta=\theta$, $X_i$,  with hazard rate $\theta$ given by the following stochastic representation,
\begin{eqnarray}
X_i|\Theta=\theta &\sim& Exp(\theta),\;\;iid,\;\;i=1,\dots,n\label{Model_11}\\
\Theta &\sim& F_\Theta(\cdot),\label{Model_12}
\end{eqnarray}
where $Exp(\theta)$ represents an exponential distribution with mean $\frac{1}{\theta}$. The joint conditional survival function is given by,
\begin{equation*}
\Pr(X_1>x_1,\dots,X_n>x_n|\Theta=\theta)=\prod_{i=1}^n\exp(-\theta x_i),\;\;x_i>0,\;i=1,2,\dots,n.
\end{equation*}
The random variable $\Theta$ represents a common random hazard rate shared by all the components or a frailty random variable in a survival context. It should be worth noting that because the components of the vector are conditionally independent and identically distributed they are exchangeable. The following characterization Theorem was provided by Albrecher et al. (2011).

\begin{theorem}
The model $(X_1,\dots,X_n)$ defined in (\ref{Model_11})-(\ref{Model_12}) can be characterized by having marginal claims $X_i$, $i=1,2,\dots,n$ that are completely monotone with a dependence structure due to an Archimedean copula with generator $\phi(u)=L_{\Theta}(u)^{-1}$ for each subset $(X_{j_1},\dots,X_{j_n})$ for $j_1,\dots,j_n$ pairwise different, where $L_{\Theta}(u)$ is the Laplace transform of $F_\Theta$.
\end{theorem}

The proof of this result can be found in Oakes (1989) and Albrecher et al. (2011).

The unconditional distribution of joint survival function is given by,
\begin{eqnarray*}
\Pr(X_1>x_1,\dots,X_n>x_n)&=&\int_0^\infty e^{-\theta(x_1+\dots+x_n)}dF_{\Theta}(\theta)\nonumber\\
&=&L_{\Theta}(x_1+\dots+x_n),
\end{eqnarray*}
for each $n$ and $x_1,\dots,x_n>0$. On the other hand, the survival joint function $\bar C$ can be written as,
$$
\Pr(X_1>x_1,\dots,X_n>x_n)=\bar C(\bar F_X(x_1),\dots,\bar F_X(x_n)),
$$
where $\bar F_X(x_i)=\Pr(X_i>x)$, is the marginal survival function, and note that all the marginal distributions are identically distributed. Now, if the survival copula is Archimedean with function generator $\phi$, then,
\begin{equation}\label{GeneralCopula}
\bar C(\bar F_{X}(x_1),\dots,\bar F_{X}(x_n))=\phi^{-1}\left(\phi(\bar F_{X_1}(x_1)+\dots+\phi(\bar F_{X_n}(x_n)\right),
\end{equation}
because the survival function of the marginal distributions are,
\begin{equation}\label{MarginalBasicModel}
\bar F_{X}(x_i)=\int_0^\infty e^{-\theta x_i}dF_\Theta(\theta)=L_{\Theta}(x_i),\;\;i=1,2,\dots,n,
\end{equation}
and then $L_{\Theta}(t)=\phi^{-1}(t)$, where $\phi(\cdot)$ is the copula's generator. The inverse of a Laplace transform of a cdf, $\phi$, is continuous strictly decreasing function from $[0,1]$ to $[0,\infty]$ with $\phi(0)=\infty$ and $\phi(1)=0$ and $\phi^{-1}$ is completely monotone. The Archimedean copula is, therefore, well defined for all $n$ and from (\ref{MarginalBasicModel}) the marginal distributions are completely monotone (Nelsen, 1999, Theorem, 4.6.2).

An alternative writing of the model (\ref{Model_11})-(\ref{Model_12}) is the following stochastic representation in terms of quotients of random variables,
\begin{equation}\label{StochasticRepresentation}
(X_1,\dots,X_n)^\top=\left(\frac{Y_1}{\Theta},\dots,\frac{Y_n}{\Theta}\right)^\top,
\end{equation}
where $Y_i$, $i=1,2,\dots,n$ are iid exponential distributions with mean 1 and $\Theta$ is a positive random variable independent of the $Y_i$. This representation is especially useful for the simulation of samples of $(X_1,\dots,X_n)^\top$, and hence to compute the different features of $S_n=\sum_{i=1}^nX_i$ in an approximate way.

\subsection{Dependence Conditions}

The dependence conditions of this model have been studied by Lee and Gross (1989) and Whitmore and Lee (1991). The multivariate distribution generated by the mixture of exponentials is dependent by total positivity of order 2 (TP2) in each pair of arguments, being the remaining arguments fixed (Barlow and Proschan, 1981). We now prove the concept of associated random variables proposed by Esary et al. (1967) in a simple way using (\ref{StochasticRepresentation}).

\begin{definition}
Random variables $X_1,\dots,X_k$ are said to be associated if
$$
cov(\phi(X_1,\dots,X_k),\psi(X_1,\dots,X_k))\ge 0
$$
for all increasing functions $(\phi,\psi)$ for which the covariance exists.
\end{definition}

The following proposition allows to check previous condition (Esary et al., 1967).
\begin{proposition}\label{Proposition1}
If,
$$
X_i=\phi_i(Y_1,\dots,Y_n),\;\;i=1,2,\dots,k
$$
where $\phi_i$ are increasing functions and $Y_1,\dots,Y_n$ are independent, then $X_1,\dots,X_k$ are associated.
\end{proposition}

We have the following result.
\begin{proposition}
The random variables $(X_1,\dots,X_n)^\top$ defined in (\ref{Model_11})-(\ref{Model_12}) are associated.
\end{proposition}
\begin{proof}
The proof is direct taking into account Proposition \ref{Proposition1} and the stochastic representation (\ref{StochasticRepresentation}).
\end{proof}

\subsection{Joint moments and dependence measures}

The joint moments of $(X_1,\dots,X_n)$ can be obtained using the following expression (Whitmore and Lee, 1991),
\begin{eqnarray}
E(X_1^{r_1}\cdots X_n^{r_n})&=&E\left(E(X_1^{r_1}\cdots X_n^{r_n}|\Theta)\right),\nonumber\\
&=&\prod_{j=1}^n\Gamma(r_j+1)E\left(\Theta^{-(r_1+\dots+r_n)}\right).\label{JointMoments}
\end{eqnarray}

On the other hand, taking into account the Archimedean character of $(X_1,\dots,X_n)$ we can obtain easily some correlations coefficients. Using (\ref{JointMoments}), Whitmore and Lee (1991) have obtained a simple expression for the linear correlation coefficient between pairs $(X_i,X_j)$. Setting $W=\Theta^{-1}$, we have
\begin{equation}\label{linealcorrelation}
\rho_{ij}=\rho(X_i,X_j)=\frac{E(W^2)-E^2(W)}{2E(W^2)-E^2(W)},
\end{equation}
and $E(W^2)\ge E^2(W)$, then $\rho_{ij}\ge 0$, which is a direct consequence of being TP2 and associated random variables.

The Kendall's tau, $\tau_{ij}$, for each bivariate pair of variable can be obtained easily from the generator function of the copula $\phi(\cdot)$. If $\frac{\phi(0)}{\phi'(0)}=0$ we have (Genest and MacKay, 1986),
\begin{eqnarray}\label{taukendall}
\tau_{ij}=\tau(X_i,X_j)=1+4\int_0^1\frac{\phi(t)}{\phi'(t)}dt.
\end{eqnarray}

\section{The distribution of the aggregated risk}\label{section3}

In this section we obtain the probability density function and the survival function of the aggregated risk $S_n$ and closed expressions for the raw moments.

\subsection{Basic result}\label{Sectionbasicresult}

\begin{theorem}\label{Theoremmain}
Let $\Theta$ be a positive random variable with cdf $F_\Theta(\cdot)$ and Laplace transform $L_{\Theta}(\cdot)$. Assume that, given $\Theta=\theta$, the random variables $X_1,\dots,X_n$ are conditionally independent and distributed as exponential $Exp(\theta)$, according to model (\ref{Model_11})-(\ref{Model_12}). Then, the pdf of the aggregated random variable $S_n=X_1+\dots+X_n$ is given by,
\begin{equation}\label{mainformula}
\displaystyle f_{S_n}(x)=\frac{x^{n-1}}{\Gamma(n)}\left\{(-1)^n\frac{d^n}{dx^n}L_{\Theta}(x)\right\},\;\;x\ge 0
\end{equation}
and $f_{S_n}(x)=0$ if $x<0$.
\end{theorem}
\begin{proof}
The unconditional distribution of $S_n$ is,
\begin{equation}\label{IntegralSn}
\displaystyle  f_{S_n}(x)=\int_0^\infty f_{S_n|\Theta}(x|\theta)dF_\Theta(\theta).
\end{equation}
Since the conditional distribution $S_n|\Theta\sim {\cal G}a(n,\theta)$ is a classical gamma distribution we have,
\begin{eqnarray*}
 f_{S_n}(x)&=&\int_0^\infty \frac{\theta^nx^{n-1}e^{-\theta x}}{\Gamma(n)}dF_\Theta(\theta) \\
 &=&\frac{x^{n-1}}{\Gamma(n)}\int_0^\infty \theta^ne^{-\theta x}dF_\Theta(\theta)\\
 &=&\frac{x^{n-1}}{\Gamma(n)}\left\{(-1)^n\frac{d^n}{dx^n}L_{\Theta}(x)\right\},
\end{eqnarray*}
and we have the result.
\end{proof}\\

Theorem \ref{Theoremmain} provides a new and simple way to obtain the probability density function of the aggregated risk in the case of dependence, using the successive derivatives of the Laplace transform. Note that $f_{S_n}(x)$ is well defined since $(-1)^kd^kL_\Theta(x)/dx^k\ge 0$, for all $k\ge 1$, because $L_\Theta(x)$ is a Laplace transform. Hence, we have three ways to compute the distribution of the aggregated risk $S_n$:
\begin{enumerate}
\item To use the stochastic representation (\ref{StochasticRepresentation}), and to obtain the pdf of $S_n=\frac{\sum_{j=1}^nY_j}{\Theta}$, as the quotient of two independent random variables.
\item To compute directly the integral (\ref{IntegralSn}).
\item To use equation (\ref{mainformula}), computing the $n$th derivative of the Laplace transform.
\end{enumerate}

\subsection{Survival function}

The survival function of the distribution of $S_n$ can be written in a simple way, using the cdf of a gamma distribution, where the shape parameter is an integer number. If $x>0$, we have,
\begin{eqnarray*}
\displaystyle\Pr(S_n>x)&=&\int_0^\infty \Pr(S_n>x|\Theta=\theta)dF_\Theta(\theta)\\
\displaystyle&=&\int_0^\infty \sum_{k=1}^{n-1}\frac{(\theta x)^ke^{-\theta x}}{k!}dF_\Theta(\theta)\\
\displaystyle&=&\sum_{k=1}^{n-1}\frac{x^{k}}{k!}\left\{(-1)^k\frac{d^k}{dx^k}L_{\Theta}(x)\right\}.
\end{eqnarray*}

Then, the survival function can be computed using the first $n-1$ first derivatives of the Laplace function.

\subsection{Moments}

The moments of $S_n$ can be obtained easily in terms of the negative moments of $\Theta$. We have,
$$
E(S_n^r)=\frac{\Gamma(n+r)}{\Gamma(n)}E(\Theta^{-r}),
$$
if $E(\Theta^{-r})<\infty$. The mean and variance of $S_n$ can be computed using the formulas,
\begin{eqnarray}
E(S_n)&=&nE(\Theta^{-1}),\label{MeanSn}\\
var(S_n)&=&nE(\Theta^{-2})+n^2var(\Theta^{-1})\label{VarianceSn},
\end{eqnarray}
assuming $E(\Theta^{-2})<\infty$.

\subsection{Value at risk}

Value at risk at level $\alpha$, with $0<\alpha<1$ of a random variable $X$ with cdf $F(x)$ is defined as,
$$\mbox{VaR}[X;\alpha]=\inf\{x\in\mathbb{R},\;F(x)\ge \alpha\}.$$
In our case, the VaR of the aggregated distribution $S_n$ is the only solution in $x$ of the equation,
$$\sum_{k=1}^{n-1}\frac{x^{k}}{k!}\left\{(-1)^k\frac{d^k}{dx^k}L_{\Theta}(x)\right\}=1-\alpha.$$
This equation can be solved numerically.

\section{Models}\label{section4}

In this section, we will study five specific mixtures of exponential distributions. We will consider dependent models with different copulas and claims of the type Pareto, classical gamma, Weibull with shape parameter $\alpha=\frac{1}{2}$, general Weibull and inverse Gaussian mixture distributions.

\subsection{Pareto Claims and Clayton Copula Dependence}

Assume that the random variable $\Theta\sim {\cal G}a(\alpha,\beta)$ is distributed as a gamma distribution with pdf,
$$
f_\Theta(\theta)=\frac{\beta^\alpha\theta^{\alpha-1}e^{-\beta\theta}}{\Gamma(\alpha)},\;\;\theta>0
$$
and Laplace transform $L_\Theta(t)=(1+t/\beta)^{-\alpha}$. The generator function of the Archimedean copula is $\phi(t)=t^{-1/\alpha}-1$. The marginal distribution $X_i$, $i=1,2,\dots,n$ are,
\begin{equation}
\bar F_{X_i}(x)=L_\Theta(x)=\frac{1}{(1+x/\beta)^\alpha},\;\;x\ge 0,\;\;i=1,2,\dots,n,\label{plt}
\end{equation}

which corresponds to a Pareto distribution ${\cal P}a(\alpha,\beta)$. The joint survival function of $(X_1,\dots,X_n)$ is,
$$
\displaystyle \Pr(X_1>x_1,\dots,X_n>x_n)=\frac{1}{\left(1+\sum\limits_{i=1}^nx_i/\beta\right)^\alpha},\;\;x_i\ge 0,\;\;i=1,2,\dots,n,
$$
which is the joint survival function of a Pareto II distribution proposed by Arnold (1983, 2015). Using (\ref{GeneralCopula}), the survival copula associated is,
\begin{equation}\label{ClaytonCopula}
\bar C(u_1,\dots,u_n)=(u_1^{-1/\alpha}+\dots+u_n^{-1/\alpha}-n+1)^{-\alpha},
\end{equation}
which is a Clayton copula. The dependence increases with $\alpha$, being the independence case obtained with $\alpha\to 0$ and the Fr\'echet upper bound when $\alpha\to\infty$. The distribution of the sum is given in the following Theorem.

\begin{theorem}
Let consider the dependent risk model $(X_1,\dots,X_n)$, with Pareto marginals with shape parameter $\alpha$, scale parameter $\beta$ and Clayton survival copula defined in (\ref{ClaytonCopula}). Then, the pdf of the aggregated risk $S_n$ is given by,
\begin{equation}\label{SumParetoDependent}
\displaystyle f_{S_n}(x)=\frac{x^{n-1}}{\beta^nB(n,\alpha)(1+x/\beta)^{n+\alpha}},\;\;x\ge 0,
\end{equation}
and $f_{S_n}(x)=0$ is $x<0$.
\end{theorem}
\begin{proof}
Since the Laplace transform of $\Theta$ is given by (\ref{plt}), using (\ref{mainformula}) we have,
$$
\displaystyle f_{S_n}(x)=\frac{x^{n-1}}{\Gamma(n)}\left\{(-1)^n\frac{d^n}{dx^n}\frac{1}{(1+x/\beta)^\alpha}\right\},
$$
and then,
\begin{eqnarray*}
L^{(n)}_\Theta(x)&=&\frac{(-1)^n\alpha(\alpha-1)\cdots (\alpha-n+1)}{\beta^n(1+x/\beta)^{\alpha+n}}\\
&=&\frac{(-1)^n\Gamma(\alpha+n)}{\beta^n\Gamma(\alpha)(1+x/\beta)^{\alpha+n}},
\end{eqnarray*}
and we obtain the result.
\end{proof}

Deploying on different methodologies, this formula was obtained by Guill\'en et al. (2013), Dacarogna et al. (2015) and Sarabia et al. (2016). As the pdf (\ref{SumParetoDependent}) is a second kind beta distribution $S_n\sim {\cal B}2(\alpha,n,\beta)$ (see McDonald, 1984). The raw moments are,
$$
\displaystyle E(S_n^r)=\frac{\beta^r\Gamma(n+r)\Gamma(\alpha-r)}{\Gamma(n+\alpha)},\;\;\mbox{if}\;\;\alpha>r.
$$


\subsection{Dependent Gamma Claims}

Our next model is based on gamma claims. The gamma distribution with shape parameter $\alpha\in(0,1]$ is completely monotone and then, it can be accommodated to the general model introduced in Section \ref{section2}. We have the following theorem (Gleser, (1989) and Albrecher and Kortschak (2009)).
\begin{theorem}
Let $X\sim {\cal G}a(\alpha,\lambda)$ be a gamma distribution with scale parameter $\lambda$ and shape parameter $\alpha\in(0,1]$ and pdf,
\begin{equation*}
f_X(x)=\frac{\lambda^\alpha x^{\alpha-1}e^{-\lambda x}}{\Gamma(\alpha)},\;\;x> 0.
\end{equation*}
Then,
$$
f_X(x)=\int_0^\infty \theta e^{-\theta x}f_\Theta(\theta)d\theta,
$$
where
\begin{equation}\label{mixinggamma}
f_\Theta(\theta)=\frac{(\theta-\lambda)^{-\alpha}\lambda^\alpha}{\theta\Gamma(1-\alpha)\Gamma(\alpha)},\;\;\lambda\le \theta<\infty,
\end{equation}
and $f_\Theta(\theta)=0$ otherwise.
\end{theorem}

The following lemma provides the Laplace transform of the mixing density (\ref{mixinggamma}).
\begin{lemma}
The Laplace transform of the random variable $\Theta$ with pdf (\ref{mixinggamma}) is,
\begin{equation}\label{LaplaceGammaClaims}
L_\Theta(s)=\frac{\Gamma(\alpha,\lambda s)}{\Gamma(\alpha)},\;\;s\ge 0,
\end{equation}
where $\Gamma(s,x)=\int_x^\infty t^{s-1}e^{-t}dt$ denotes the upper incomplete gamma function.
\end{lemma}
\begin{proof}
The proof is direct using the pdf defined in (\ref{mixinggamma}).
\end{proof}

Using previous lemma, we get the generator function of the corresponding copula, which is given by,
\begin{equation}\label{GeneratorGammaClaims}
\phi(t)=Q_{G_\alpha}(1-t),
\end{equation}
where $Q_{G_\alpha}(u)$ represents the quantile function of a gamma distribution with mean $\alpha$ and unit scale parameter. Using Equation (\ref{LaplaceGammaClaims}), the joint survival function is,
$$
\displaystyle \Pr(X_1>x_1,\dots,X_n>x_n)=\frac{\Gamma\left(\alpha,\lambda\sum_{j=1}^nx_j\right)}{\Gamma(\alpha)},
$$
if $x_1,\dots,x_n\ge 0$, with marginal survival functions,
$$
\displaystyle \Pr(X_i>x)=\frac{\Gamma\left(\alpha,\lambda x\right)}{\Gamma(\alpha)},\;\;x\ge 0,\;i=1,2,\dots,n.
$$
The associated copula is given in the following Theorem.
\begin{theorem}
The survival copula associated to the Exponential-Gamma dependent model is given by,
\begin{equation}\label{GammaCopula}
\bar C(u_1,\dots,u_n)=1-F_{G_\alpha}\left(Q_{G_\alpha}(1-u_1)+\dots+Q_{G_\alpha}(1-u_2)\right),
\end{equation}
where $F_{G_\alpha}(\cdot)$ and $Q_{G_\alpha}(\cdot)$ represent the cdf and the quantile function, respectively of the gamma distribution with shape parameter $\alpha$.
\end{theorem}
\begin{proof}
The proof is direct by considering equation (\ref{GeneralCopula}) and the generator function (\ref{GeneratorGammaClaims}).
\end{proof}

\begin{theorem}
Let consider the dependent risk model $(X_1,\dots,X_n)$, with gamma marginals with shape parameter $\alpha\in (0,1)$, scale parameter $\lambda$ and survival copula defined in (\ref{GammaCopula}). Then, the pdf of the aggregated risk $S_n$ is given by,
\begin{equation}\label{SumGammaClaims}
\displaystyle f_{S_n}(x)=\sum_{k=0}^{n-1}\frac{(-1)^k(\alpha-1)_k}{\Gamma(\alpha)k!(n-k-1)!}\lambda(\lambda x)^{n+\alpha-k-2}e^{-\lambda x},\;\;x\ge 0
\end{equation}
with $n=2,3,\dots$, $f_{S_n}(x)=0$ if $x<0$ and $(a)_n=a(a-1)\cdots (a-n+1)$ is the Pochhammer symbol. Previous pdf can be written as a finite mixture of Gamma distributions ${\cal G}a(\alpha_k,\lambda)$,
$$
f_{S_n}(x)=\sum_{k=0}^{n-1}w_kf_{{\cal G}a(\alpha_k,\lambda)}(x),
$$
with shape parameters,
\begin{equation}\label{akGamma}
\alpha_k=n+\alpha-k-1,\;\;k=0,1,\dots,n-1,
\end{equation}
and weights (positive and negatives)
\begin{equation}\label{weightsGamma}
w_k=\frac{(-1)^k(\alpha-1)_k\Gamma(n+\alpha-k-1)}{\Gamma(\alpha)k!(n-k-1)!},\;\;k=0,1,\dots,n-1.
\end{equation}

\end{theorem}
\begin{proof}
Since $L_\Theta'(x)=-\frac{\lambda^\alpha}{\Gamma(\alpha)}x^{\alpha-1}e^{-\lambda x}$ and using the Leibniz's rule (see Appendix, equation (\ref{LeibnizFormula})) we have,
\begin{eqnarray*}
\frac{d^n}{dx^n}L_\Theta(x)&=&\frac{d^{n-1}}{dx^{n-1}}\left\{L_\Theta'(x)\right\}\\
&=&-\frac{\lambda^\alpha}{\Gamma(\alpha)}\frac{d^{n-1}}{dx^{n-1}}\left\{e^{-\lambda x}x^{\alpha-1}\right\}\\
&=&-\frac{\lambda^\alpha}{\Gamma(\alpha)}\sum_{k=0}^{n-1}{n-1\choose k}(-1)^{n-k-1}\lambda^{n-k-1}e^{-\lambda x}(\alpha-1)_kx^{\alpha-k-1},
\end{eqnarray*}
and using (\ref{mainformula}) we have,
$$
f_{S_n}(x)=\sum_{k=0}^{n-1}\frac{(-1)^{2n-k}(\alpha-1)_k}{\Gamma(\alpha)\Gamma(n)}{n-1\choose k}\lambda^{n+\alpha-k-1}x^{n+\alpha-k-2}e^{-\lambda x},
$$
which corresponds to formula (\ref{SumGammaClaims}).
\end{proof}\\

The moments of (\ref{SumGammaClaims}) can be obtained easily taking into account that it is a finite mixture of Gamma distributions. We have that,
\begin{equation*}
E(S_n^r)=\sum_{k=0}^{n-1}w_k\frac{\Gamma(\alpha_k+r)}{\lambda^r\Gamma(\alpha_k)},
\end{equation*}
where $r>0$, $\alpha_k$ are defined in (\ref{akGamma}) and $w_k$ in (\ref{weightsGamma}).

\subsection{Weibull $\frac{1}{2}$ Claims with Gumbel Copula Dependence}

This model corresponds to Weibull Claims with Gumbel copula dependence (Albrecher et al., 2011). We consider for $\Theta$ a $\frac{1}{2}$-stable distribution, also called L\'evy distribution, with probability density function (see Jewell, 1982),
\begin{equation}\label{PdfLevy}
f_\Theta(\theta)=\frac{\lambda}{2\sqrt{\pi \theta^3}}e^{-\lambda^2/4\theta},\;\;\theta\ge 0,
\end{equation}
and Laplace transform $L_\Theta(s)=e^{-\lambda\sqrt{s}}$, which corresponds to the function generator of the Archimedean copula $\phi(s)=(-\log(s))^2$, thus being a special case of the Gumbel Copula. For this model, the corresponding survival marginal function,
$$
\bar F_{X_i}(x)=\int_0^\infty e^{-\theta x}f_\Theta(\theta)d\theta=\exp(-\lambda \sqrt{x}),\;\;x\ge 0.
$$
The joint survival function is,
$$
\Pr(X_1>x_1,\dots,X_n>x_n)=\exp\left(-\lambda\sqrt{x_1+\dots+x_n}\right),\;\;x_1\ge 0,\dots,x_n\ge 0.
$$
The following Theorem states a simpler expression to the formulation provided by Dacarongna et al. (2015). To obtain this result, we use of the probability density function of the generalized inverse Gaussian (GIG) distribution introduced by Good (1953), with pdf,
\begin{equation}\label{InverseGaussianpdf}
f(x;a,b,p)=\frac{(a/b)^{p/2}}{2K_p(\sqrt{ab})}x^{p-1}\exp\left\{-\frac{1}{2}\left(ax+\frac{b}{x}\right)\right\},\;\;x\ge 0,
\end{equation}
where $-\infty<p<\infty$, $(a,b)\in\Theta_p$, where $\Theta_p=\{(a,b):\;a>0,\;b\ge 0\}$ if $p>0$, $\{(a,b):\;a>0,\;b>0\}$ if $p=0$ and $\{(a,b):\;a\ge 0,\;b>0\}$ if $p<0$. Here, $K_\nu(z)$ denotes the modified Bessel function of the third kind with index $\nu$ and argument $z$ (Watson, 1995). Special sub-models are the gamma distribution ($b=0$, $p>0$), the reciprocal gamma distribution ($a=0$, $p<0$), the inverse Gaussian distribution ($p=-1/2$) and the hyperbola distribution ($p=0$).

\begin{theorem}
Let consider the dependent risk model $(X_1,\dots,X_n)$, with marginal Weibull distributions with shape parameter $\alpha=\frac{1}{2}$ and scale parameter $c>0$ and Gumbel survival copula with dependent parameter $\theta=1/\alpha=2$. Then, the pdf of the aggregated risk $S_n$ is given by,
\begin{equation}\label{SumWeibull2}
f_{S_n}(x)=\frac{\lambda}{2^{2n-1}\Gamma(n)}\sum_{k=0}^{n-1}\frac{(2(n-1)-k)!}{(n-k-1)!k!}(2\lambda)^kx^{(k-1)/2}e^{-\lambda\sqrt{x}},\;\;x\ge 0,
\end{equation}
and $f_{S_n}(x)=0$ if $x<0$.
\end{theorem}
\begin{proof}
To prove this result, we use the second method described in Section \ref{Sectionbasicresult}. Using (\ref{IntegralSn})and (\ref{PdfLevy}) we have,
\begin{eqnarray*}
f_{S_n}(x)&=&\frac{x^{n-1}}{\Gamma(n)}\int_0^\infty t^{n}e^{-xt}f_\Theta(t)dt\\
&=&\frac{\lambda x^{n-1}}{2\sqrt{\pi}\Gamma(n)}\int_0^\infty t^{n-3/2}e^{-xt-\lambda^2/4t}dt\\
&=&\frac{\lambda x^{n-1}}{2\sqrt{\pi}\Gamma(n)}\frac{2K_{n-1/2}(\lambda\sqrt{x})}{(4x/\lambda^2)^{(1/2)(n-1/2)}}\\
&=&\frac{\lambda^{n+1/2}}{2^{n-1/2}\sqrt{\pi}\Gamma(n)}x^{n/2-3/4}K_{n-1/2}(\lambda\sqrt{x}),
\end{eqnarray*}
where we have s GIG distribution (\ref{InverseGaussianpdf}) with parameters $a=2t$, $b=\lambda^2/2$ and $p=n-1/2$. Now, we use the following result by Gradshteyn and Ryzkiz (1980),
$$
K_{n+\frac{1}{2}}(x)=\frac{\sqrt{\pi}}{(2x)^{n+1/2}}e^{-x}\sum_{k=0}^n\frac{(2n-k)!}{(n-k)!k!}(2x)^k,
$$
and after some computation we obtain (\ref{SumWeibull2}).
\end{proof}

Formula (\ref{SumWeibull2}) can also be written in terms of the partial Bell polynomials (see Appendix). We write $L_\Theta(x)=f(g(x))=e^{-\lambda\sqrt{x}}$, with $f(x)=e^{-\lambda x}$, $g(x)=\sqrt{x}$. Since $f^{(n)}(x)=(-1)^n\lambda^ne^{-\lambda x}$ and
$$
g^{(n)}(x)=a_nx^{1/2-n},
$$
where
\begin{equation}\label{ancoefficients}
a_n=\frac{(-1)^{n-1}(2n-2)!}{2^{2n-1}(n-1)!},\;\;n=1,2,\dots
\end{equation}
we have,
\begin{equation*}
f_{S_n}(x)=\frac{x^{n-1}}{\Gamma(n)}\sum_{k=1}^n(-1)^{n-k}\lambda^ke^{-\lambda\sqrt{x}}B_{n,k}\left\{a_jx^{1/2-j},\;\;1\le j\le n-k+1\right\},
\end{equation*}
with $x\ge 0$, being $B_{n,k}(x_1,\dots,x_{n-k+1})$ are the partial Bell polynomials defined in the Appendix.

Note that the density (\ref{SumWeibull2}) is a finite mixture of densities of the form $f(x)\propto x^{\alpha-1}e^{-\lambda\sqrt{x}}$. We introduce the following definition.
\begin{definition}\label{DefinitionSquareGamma}
A random variable $X$ is said to have a square gamma distribution if its pdf is of the form,
\begin{equation}\label{pdfSquareGamma}
f(x)=\frac{\lambda^{2\alpha}x^{\alpha-1}e^{-\lambda\sqrt{x}}}{2\Gamma(2\alpha)},\;\;x\ge 0,
\end{equation}
and $f(x)=0$ is $x<0$, with $\alpha,\lambda>0$.
\end{definition}
A random variable with pdf (\ref{pdfSquareGamma}) will be represented as $X\sim {\cal SG}a(\alpha,\lambda)$. Note that $\sqrt{X}\sim {\cal G}a(2\alpha,r)$, that is, the square root of $X$ is a classical gamma distribution. The raw moments are $E(X^r)=\frac{\Gamma(2(\alpha+r))}{\lambda^{2r}\Gamma(2\alpha)}$, with $r>0$. The mean and variance of $\sqrt{X}$ are $E(\sqrt{X})=\frac{2\alpha}{\lambda}$ and $var(\sqrt{X})=\frac{2\alpha}{\lambda^2}$ respectively.

Using Definition \ref{DefinitionSquareGamma}, the density of $S_n$ can be written as,
$$
f_{S_n}(x)=\sum_{k=0}^{n-1}\frac{(2n-k-2)!}{(n-k-1)!\Gamma(n)2^{2n-k-2}}\frac{\lambda^{k+1}x^{(k+1)/2-1}e^{-\lambda\sqrt{x}}}{2\Gamma(k+1)},
$$
which is a finite mixture of square gamma distributions with components $X_k\sim {\cal SG}a(\frac{k+1}{2},\lambda)$, $k=1,2,\dots,n$,
$$
f_{S_n}(x)=\sum_{k=0}^{n-1}w_kf_{{\cal SG}a((k+1)/2,\lambda)}(x),
$$
and weights,
$$
w_k=\frac{(2n-k-2)!}{(n-k-1)!\Gamma(n)2^{2n-k-2}},\;\;k=0,1,\dots,n-1.
$$
The moments of the sum $S_n$ can be obtained in a simple form as,
$$
E(S_n^r)=\sum_{k=0}^{n-1}w_kE\left(X_k^r\right)=\sum_{k=0}^{n-1}w_k\frac{\Gamma(k+1+2r)}{\lambda^{k+1}\Gamma(k+1)}.
$$

\subsection{General Weibull Claims with Gumbel Copula Dependence}

We introduce in this section a new model for the aggregated risks, which is an extension of the previous one. Let consider a positive stable random variable with pdf (see Feller, 1971),
$$
f_\Theta(x)=-\frac{1}{\pi x}\sum_{k=1}^\infty \frac{\Gamma(k\alpha+1)}{k!}(-x^{-\alpha})^k\sin(\alpha k\pi),
$$
and Laplace transform,
\begin{equation}\label{laplacestable}
L_\Theta(s)=\exp(-s^\alpha),\;\;s\ge 0,
\end{equation}
and $\alpha\in(0,1]$. Using (\ref{laplacestable}), the joint survival function is,
$$
\Pr(X_1>x_1,\dots,X_n>x_n)=\exp\{-(x_1+\dots+x_n)^\alpha\},
$$
with marginal distributions,
$$
\bar F_{X_i}(x)=\exp\{-x^\alpha\},\;\;x\ge 0,\;\;i=1,2,\dots,n,
$$
which are Weibull distributions with shape parameter $\alpha\in(0,1]$. Since the Laplace transform of $\Theta$ is (\ref{laplacestable}), the generator of the Archimedean copula is $\phi(t)=L_\Theta^{-1}(t)=(-\log t)^{1/\alpha}$ and the survival copula is,
\begin{equation}\label{cgumbel}
\bar C_\alpha(u_1,\dots,u_n)=\exp\left\{-\left[\sum_{k=1}^n(-\log u_k)^{1/\alpha}\right]^\alpha\right\},
\end{equation}
which corresponds to the Gumbel copula (the family B6 in the Joe's (1997) notation). Note that (\ref{cgumbel}) includes the $C_U$ upper bound for $\alpha\to 0$ and the $C_I$ bound for $\alpha=1$.

\noindent For the distribution of the sum, we have the following Theorem.
\begin{theorem}
Let consider the model (\ref{Model_11})-(\ref{Model_12}), where the marginal claims are Weibull($\alpha$) with $0<\alpha\le 1$ and the copula dependence is Gumbel. Then, the distribution of $S_n$ is,
\begin{equation}\label{sumWeibull}
\displaystyle f_{S_n}(x;\alpha)=\frac{x^{n-1}}{\Gamma(n)}\sum_{k=1}^n(-1)^{n+k}e^{-x^\alpha}B_{n,k}\left((\alpha)_1x^{\alpha-1},\dots,(\alpha)_{n-k+1}x^{\alpha-(n-k+1)}\right),
\end{equation}
where $\alpha\in (0,1]$, $B_{n,k}(x_1,\dots,x_{n-k+1})$ are the partial Bell polynomials and $(\alpha)_n$ is the Pochhammer symbol.
\end{theorem}
\begin{proof}
We write $L_\Theta(x)=f(g(x))=e^{-t^\alpha}$, where $f(x)=e^{-x}$ and $g(x)=x^\alpha$. Then,
$$
f^{(n)}(x)=(-1)^ne^{-x},
$$
and
$$
g^{(n)}(x)=\alpha(\alpha-1)\cdots(\alpha-n+1)x^{\alpha-n}=(\alpha)_nx^{\alpha-n},
$$
so, using the Fa\`a di Bruno formula (\ref{FormulaFaadiBruno}) and Theorem \ref{Theoremmain}, we get the result.
\end{proof}\\

The pdf (\ref{sumWeibull}) is a finite mixture of generalized gamma distributions $GG(\alpha,\eta)$, whose pdf is (see McDonald, 1984),
$$
f(x;\alpha,\eta)=\frac{\alpha x^{\alpha\eta-1}e^{-x^\alpha}}{\Gamma(\alpha)},\;\;x\ge 0.
$$
For $n=2$, formula (\ref{sumWeibull}) becomes in,
$$
f_{S_2}(x;\alpha)=(1-\alpha)\alpha x^{\alpha-1}e^{-x^\alpha}+\alpha^2x^{2\alpha-1}e^{-x^\alpha},\;\;x\ge 0,
$$
which is a finite mixture of two generalized gamma $GG(\alpha,1)$ and $GG(\alpha,2)$, with weights $1-\alpha$ and $\alpha$ respectively. For $n=3$ we have,
$$
f_{S_3}(x;\alpha)=\frac{\alpha(1-\alpha)(2-\alpha)}{2} x^{\alpha-1}e^{-x^\alpha}+\frac{3\alpha^2(1-\alpha)}{2}x^{2\alpha-1}e^{-x^\alpha}+\frac{\alpha^3}{2}x^{3\alpha-1}e^{-x^\alpha},
$$
with $x\ge 0$, which is again a mixture of the three generalized gamma distributions $GG(\alpha,j)$, with $j=1,2,3$ with weights $\frac{(1-\alpha)(2-\alpha)}{2}$, $\frac{3\alpha(1-\alpha)}{2}$ and $\alpha^2$.\\

The following Lemma provides the Kendall's tau coefficient. 
\begin{lemma}
The tau Kendall's tau coefficient between pairs of random variables is,
$$
\tau_{ij}(X_i,X_j)=1-\alpha.
$$
\end{lemma}
\begin{proof}
Since $\frac{\phi(0)}{\phi'(0)}=0$, using Equation (\ref{taukendall}), we obtain the result.
\end{proof}

\subsection{Inverse Gaussian Mixtures of exponential Claims}

If $\Theta\sim {\cal IG}(\lambda,\mu)$ has an inverse Gaussian distribution with parameters $\mu>0$ and $\lambda>0$ and pdf,
$$
f_\Theta(x)=\sqrt{\frac{\lambda}{2\pi}}x^{-3/2}\exp\left(-\frac{\lambda(x-\mu)^2}{2\mu^2x}\right),\;\;x>0,
$$
the corresponding mixing distribution for the marginal claim size $X_j$ is,
\begin{equation}\label{SurvivalInverseGaussian}
\Pr(X_j>x)=\int_0^\infty e^{-\theta x}f_\Theta(\theta)d\theta=\exp\left\{-\frac{\lambda}{\mu}\left(\sqrt{1+\frac{2\mu^2x}{\lambda}}-1\right)\right\},\;\;x\ge 0,
\end{equation}
for $j=1,2,\dots,n$ and joint survival function
\begin{equation}\label{JointSurvivalIG}
\Pr(X_1>x_1,\dots,X_n>x_n)=\exp\left\{-\frac{\lambda}{\mu}\left(\sqrt{1+\frac{2\mu^2}{\lambda}\sum_{j=1}^nx_j}-1\right)\right\},
\end{equation}
if $x_1,\dots,x_n\ge 0$. This model (with a different parameterization) was introduced by Whitmore (1988) and extended to the multivariate case by Whitmore and Lee (1991).
Since the generator is,
$$
\phi(t)=\frac{\lambda}{2\mu^2}\left\{\left(1-\frac{\mu}{\lambda}\log t\right)^2-1\right\},
$$
the survival copula associated is,
\begin{equation}\label{copulamixtureIG}
\bar C(u_1,\dots,u_n)=\exp\left\{-\frac{\lambda}{\mu}\left[\left(\sum_{j=1}^n\left(1-\frac{\mu}{\lambda}\log u_j\right)^2-n+1\right)^{1/2}-1\right]\right\},
\end{equation}
with $0\le u_j\le 1$, $j=1,2,\dots,n$.

To obtaine the pdf of the sum, we define
\begin{equation*}
a(x)=\sqrt{1+\frac{2\mu^2x}{\lambda}}-1,
\end{equation*}
and
\begin{equation*}
b(x)=\frac{\lambda}{\mu}a(x),
\end{equation*}
We have the following Theorem.
\begin{theorem}
Let consider the model (\ref{Model_11})-(\ref{Model_12}), where the marginal claims are defined in (\ref{SurvivalInverseGaussian}) and the copula dependence is (\ref{copulamixtureIG}). Then, the distribution of $S_n$ is,
\begin{equation*}
f_{S_n}(x)=\frac{x^{n-1}}{\Gamma(n)}\sum_{k=1}^n(-1)^{n+k}\left(\frac{\lambda}{\mu}\right)^ke^{-b(x)}u_{n,k}(x),
\end{equation*}
if $x\ge 0$, $f_{S_n}(x)=0$ if $x<0$ where
\begin{equation*}
u_{n,k}(x)=B_{n,k}\left\{a_jb^j(1+bx)^{1/2-j},\;\;1\le j\le n-k+1 \right\},
\end{equation*}
with $b=\frac{2\mu^2}{\lambda}$, $B_{n,k}(x_1,\dots,x_{n-k+1})$ are the partial Bell polynomials and $a_n$ are defined in (\ref{ancoefficients}).
\end{theorem}
\begin{proof}
We write $L_\Theta(t)=f(g(t))$, where $f(t)=e^{-\frac{\lambda}{\mu}t}$ and $g(x)=(1+b(x))^{1/2}-1$. We have $f^{(n)}(t)=(-1)^n(\frac{\lambda}{\mu})^ne^{-\frac{\lambda}{\mu}t}$ and
$$g^{(n)}(t)=a_nb^n(1+bx)^{1/2-n},$$
where $b=2\mu^2/\lambda$ and $a_n$ is defined in (\ref{ancoefficients}). Then, using the Fa\`a di Bruno formula (\ref{FormulaFaadiBruno}), we obtain the result.
\end{proof}

We have the following formulas for $n=2,3$ and $4$,
$$
f_{S_2}(x)=\frac{\mu^3xe^{-b(x)}}{\lambda(a(x)+1)^3}+\frac{\mu^2xe^{-b(x)}}{(a(x)+1)^2},\;\;x\ge 0,
$$
$$
f_{S_3}(x)=\frac{3\mu^5x^2e^{-b(x)}}{2\lambda^2(a(x)+1)^5}+
\frac{3\mu^4x^2e^{-b(x)}}{2\lambda(a(x)+1)^4}+
\frac{\mu^3x^2e^{-b(x)}}{2(a(x)+1)^3},\;\;x\ge 0,
$$
and
$$
f_{S_4}(x)=\frac{15\mu^7x^3e^{-b(x)}}{6\lambda^3(a(x)+1)^7}+
\frac{15\mu^6x^3e^{-b(x)}}{6\lambda^2(a(x)+1)^6}+
\frac{6\mu^5x^3e^{-b(x)}}{6\lambda(a(x)+1)^5}+
\frac{\mu^4x^3e^{-b(x)}}{(a(x)+1)^4},\;\;x\ge 0.
$$

The moments of $S_n$ can be obtained in this way. If $\Theta\sim {\cal IG}(\lambda,\mu)$ is an inverse Gaussian distribution, the positive moments are (see Johnson et al., 1994),
\begin{equation}\label{MomentsIG}
E(\Theta^r)=\mu^r\sum_{s=0}^{r-1}\frac{(r-1+s)!}{s!(r-1-s)!}\left(2\frac{\lambda}{\mu}\right)^{-s},
\end{equation}
and the negative moments
\begin{equation}\label{MomentsInverseIG}
E(\Theta^{-r})=\frac{E(\Theta^{r+1})}{\mu^{2r+1}},\;\;r=1,2,\dots
\end{equation}
Using previous formulas and Equations (\ref{MeanSn}) and (\ref{VarianceSn}) the mean and variance of $S_n$ are
\begin{eqnarray*}
E(S_n)&=&n\left(\frac{1}{\lambda}+\frac{1}{\mu}\right),\\
var(S_n)&=& n\left(\frac{1}{\mu^2}+\frac{3}{\lambda\mu}+\frac{3}{\lambda^2}\right)+n^2\left(\frac{1}{\lambda\mu}+\frac{2}{\lambda^2}\right).
\end{eqnarray*}
The mean and variance of the marginal claims with cdf (\ref{SurvivalInverseGaussian}) are,
$$
E(X_i)=\frac{1}{\lambda}+\frac{1}{\mu},\;\;i=1,2,\dots,n,
$$
and
$$
var(X_i)=\frac{\lambda^2+4\lambda\mu+5\mu^2}{\lambda^2\mu^2},\;\;i=1,2,\dots,n.
$$
For the multivariate distribution $(X_1,\dots,X_n)$ with joint survival function (\ref{JointSurvivalIG}), the linear correlation coefficient between pairs $(X_i,X_j)$ is (using (\ref{linealcorrelation}) and Equations (\ref{MomentsIG}) and (\ref{MomentsInverseIG})),
$$
\rho_{ij}=\rho(X_i,X_j)=\frac{\mu(\lambda+2\mu)}{\lambda^2+4\lambda\mu+5\mu^2},\;\;i\neq j.
$$
The Kendall's tau coefficient is given in the following Lemma.
\begin{lemma}
For the multivariate distribution $(X_1,\dots,X_n)$ with joint survival function (\ref{JointSurvivalIG}), the Kendall's tau coefficient for each pair $(X_i,X_j)$ is,
$$
\tau_{ij}(X_i,X_j)=1-\frac{a(2+a)-4e^{2/a}\Gamma(0,2/a)}{2a^2},\;i\neq j,
$$
where $\Gamma(0,z)$ is the incomplete gamma function and $a=\frac{\mu}{\lambda}$.
\end{lemma}
\begin{proof}
Since $\frac{\phi(0)}{\phi'(0)}=0$, using (\ref{taukendall}) and calling $a=\frac{\mu}{\lambda}$,
$$
\int_0^1\frac{\phi(t)}{\phi'(t)}dt=-\int_0^1\frac{[(1-a\log t)-1]^2-1}{2a(1-a\log t)}dt=-\int_0^\infty \frac{[(1+ax)-1]e^{-2x}}{2a(1+ax)}dx,
$$
and we get the result, after making the change of variable $\log t=x$.
\end{proof}

\subsection{Computation of risk measures and ruin formulas}

Here we discuss briefly the computation of risk measures for the aggregated distribution. For the model with Pareto claims and Clayton copula dependence, expressions for the VaR, TVar and tail moments have been obtained by Sarabia et al. (2016). On the other hand, some of the aggregated distributions obtained in previous sections can be written as finite mixtures of distributions. Then, to compute the TVaR and the tail moments, we include in the Appendix a general result for computing these measures in finite mixtures, in terms of these measures for the components of the mixture.

Explicit ruin formulas for this kind of models with dependence risks have been provided by Albrecher et al. (2011). They obtained explicit mixture of ruin functions by using the gamma and Weibull (the pdf given in (\ref{PdfLevy})) as mixing distributions. Unfortunately, we can not get additional mixture ruin functions for the others mixing distributions considered here. Notwithstanding, we can obtain a closed expression for the mixture ruin function for model (\ref{Model_11}-\ref{Model_12}) and assuming that $\Theta$ follows a Lindley distribution with parameter $\lambda>0$. The continuous Lindley distribution
(Lindley, 1958), which depends only on one parameter, has not been frequently used in the statistical science although its treatment
has been demonstrated to be useful. The pdf of the Lindley distribution with parameter $\lambda>0$ is given by
\begin{eqnarray*}
f(x)=\frac{\lambda^2}{1+\lambda}(1+x)\exp(-\lambda x),\quad x>0,
\end{eqnarray*}
which is a mixture of an exponential distribution with a gamma distribution. The use of this simple distribution has been demonstrated to be useful in actuarial statistical in some recent works. See, for example G\'omez--D\'eniz et al. (2012) and  Asgharzadeh et al. (2017).

Now, under the model (\ref{Model_11}-\ref{Model_12}) and using Theorem \ref{Theoremmain}, we get the probability function of $S_n$, which results in
\begin{eqnarray}
f_{S_n}(x)=\frac{n\lambda^2}{1+\lambda}\frac{x^{n-1}(x+\lambda+n+1)}{(x+\lambda)^{n+2}},\quad x>0.\label{cml}
\end{eqnarray}

Using expression (4) in Albrecher et al. (2011) and denoting by $\psi_{\theta}(u)$ the ruin probability of the classical compound Poisson (with parameter $\phi>0$) we get the mixture of the ruin function when $\Theta$ follows the Lindley distribution with parameter $\lambda>0$. This results
\begin{eqnarray*}
\psi(u) &=& 1-\frac{1+\lambda(1+\theta_0)}{1+\lambda}\exp(-\theta_0\lambda)\\
&&+\frac{\lambda^2\phi \exp(u\phi/c)}{c(1+\lambda)(u+\lambda)}\left[
\exp[-\theta_0(u+\lambda)]+(u+\lambda)\Gamma(0,\theta_0(u+\lambda))\right].
\end{eqnarray*}

Here, $\theta_0=\phi/c$ and $c>0$ is a constant premium intensity. Observe that when $u\to \infty$ $\psi(u)=1-\frac{1+\lambda(1+\theta_0)}{1+\lambda}\exp(-\theta_0\lambda)=\bar F(\theta_0)$, where $\bar F(\cdot)$ is the survival function of the Lindley distribution.

\subsection{Collective risk model}
Sarabia et al. (2016) obtained some closed-form expressions for the pdf of the total claim amount in the collective risk model, $S_N=X_1+X_2+\dots+X_N$, assuming that the secondary distribution is Pareto and several primary distributions for $N$. The pdf of the total claim amount can be computed by using
\begin{equation}
g_{S_N}(x)=\sum_{n=1}^\infty p_n f_{S_n}(x),\quad x>0\label{cmc}
\end{equation}
being $g_{S_N}(0)=p_0$.

It is easy to check using (\ref{cml}) and (\ref{cmc}) that if $N$ follows a Poisson distribution with parameter $\phi>0$ and the secondary distribution is the Lindley with parameter $\lambda>0$, then the total claim amount has pdf given by
\begin{eqnarray*}
g_{S_N}(x)=\frac{\lambda(\lambda+2)+x[2(\lambda+1)+\phi+x]}{(\lambda+1)(\lambda+x)^4}\phi\lambda^2 \exp\left[-\frac{\lambda  \phi }{\lambda +x}\right], \quad x>0,
\end{eqnarray*}
while $g_{S_N}(0)=\exp(-\phi)$.

On the other hand, if the primary distribution is negative binomial with parameters $r>0$ and $0<p<1$, then the pdf of the total claim amount results,
\begin{eqnarray}
g_{S_N}(x)=\frac{\lambda(\lambda+2)+x[p(x+\lambda-r+1)+\lambda+r+1]}{(\lambda+1)(\lambda+px)^{2+r}}
(x+\lambda)^{r-2}\lambda^2qrp^r,\label{cll}
\end{eqnarray}
for $x>0$, being $f_{S_N}(0)=p^r$ and $q=1-p$. Observe that if we assume in (\ref{cll}) $r=1$ we get the pdf of the total claim amount for the compound geometric-Lindley model.

Finally, if we assume the logarithmic distribution with parameter $0<\phi<1$ as primary distribution we get
\begin{eqnarray*}
g_{S_N}(x)=\frac{\lambda ^2 \phi  [x \phi  (\lambda +x+1)-(\lambda +x) (\lambda +x+2)]}{(\lambda +1) [(\lambda +x) (\lambda +x(1-\phi)]^2 \log (1-\phi )},\quad x>0,
\end{eqnarray*}
and $g_{S_N}(0)=0$ if $x<0$.

Expressions for the mean and variance of these compounds distributions can be derived easily.

\section{More General Dependent Models}\label{section5}

In this section we sketch one extension of the basic model. One of the previous extensions of the basic model (\ref{Model_11})-(\ref{Model_12}) was provided by Albrecher et al. (2011), using conditional marginals with survival function of the power form $(\bar G(x_i))^\theta$, for a particular baseline distribution $G(x)$. The new family is again Archimedean and includes the model with Pareto claims. However, the distribution of the aggregated risks does not look simple.

We consider below mixtures of classical gamma distributions. We have a random vector $(X_1,\dots,X_n)$, which are conditionally on $\Theta$ independent gamma distributions with shape parameter $\alpha_i$ and location $\theta$. The corresponding stochastic representation is,
\begin{eqnarray*}
X_i|\Theta=\theta &\sim& {\cal G}a(\alpha_i,\theta),\;\;i=1,2,\dots,n,\;\;independent\\
\Theta &\sim& F_\Theta(\cdot),
\end{eqnarray*}
with $\alpha_1,\dots\alpha_n>0$ and $\theta>0$. Taking $\alpha_i=1$ for all $i$ in previous model, we obtain (\ref{Model_11})-(\ref{Model_12})

The conditional joint density of $(X_1,\dots,X_n)$ given $\Theta$ is,
$$
\displaystyle f(x_1,\dots,x_n|\theta)=\prod_{i=1}^n\frac{x_i^{\alpha_i-1}}{\Gamma(\alpha_i)}\theta^{\tilde\alpha}e^{-\theta\sum_{i=1}^nx_i},\;\;x_i>0,\;i=1,2,\dots,n,
$$
where $\tilde\alpha=\sum_{i=1}^n\alpha_i$. The joint density of $(X_1,\dots,X_n)$ distribution is,
$$
\displaystyle f_{X_1,\dots,X_n}(x_1,\dots,x_n)=\prod_{i=1}^n\frac{x_i^{\alpha_i-1}}{\Gamma(\alpha_i)}\int_0^\infty\theta^{\tilde\alpha}e^{-\theta\sum_{i=1}^nx_i}dF_\Theta(\theta),
$$
where $x_i>0$, $i=1,2,\dots,n$, with marginal densities,
\begin{equation}\label{pdfmarginalgeneral}
\displaystyle f_{X_i}(x_i)=\frac{x_i^{\alpha_i-1}}{\Gamma(\alpha_i)}\int_0^\infty\theta^{\alpha_i}e^{-\theta x_i}dF_\Theta(\theta),
\end{equation}
$i=1,2,\dots,n$. It should be noted that $X_i$ are not equally distributed if $\alpha_i\neq\alpha_j$, $i\neq j$.

Taking into account that $S_n|\theta\sim {\cal G}a(\tilde\alpha,\theta)$ we obtain,
$$
\displaystyle f_{S_n}(x)=\frac{x^{\tilde\alpha-1}}{\Gamma(\tilde\alpha)}\int_0^\infty \theta^{\tilde\alpha}e^{-\theta x}dF_\Theta(\theta),\;\;x\ge 0,
$$
and $f_{S_n}(x)=0$ if $x<0$.

\subsection{Second kind beta mixtures of gamma distributions}

Assume the the mixing distribution is a second kind beta distribution, and then we consider the multivariate dependent risk model,
\begin{equation}\label{MultivariateSibuya}
(X_1,\dots,X_n)^\top=(G_{\alpha_1}H_{\beta,\gamma},\dots,G_{\alpha_n}H_{\beta,\gamma})^\top,
\end{equation}
where $G_{\alpha_i}\sim {\cal G}a(\alpha_i,1)$, $i=1,2,\dots,n$ are independent gamma random variables and $H_{\beta,\gamma}=\frac{G_\beta}{G_\gamma}\sim {\cal B}2(\beta,\gamma)$ is a second kind beta distribution independent of the gamma random variables. The marginal distribution of (\ref{MultivariateSibuya}) was proposed by Sibuya (1979), using the term ``gamma product-ratio distributions".  The marginal pdf of $X_i$ is given by (using (\ref{pdfmarginalgeneral})),
\begin{eqnarray*}
f_{X_i}(x)&=&\frac{\Gamma(\beta+\gamma)}{\Gamma(\alpha)\Gamma(\beta)\Gamma(\gamma)}x^{\alpha_i-1}\int_0^\infty e^{-x\theta}\frac{\theta^{\alpha_i+\gamma-1}}{(1+\theta)^{\beta+\gamma}}d\theta\\
&=&\frac{\Gamma(\alpha_i+\gamma)\Gamma(\beta+\gamma)}{\Gamma(\alpha_i)\Gamma(\beta)\Gamma(\gamma)}x^{\alpha_i-1}U(\alpha_i+\gamma,\alpha_i-\beta+1,x),
\end{eqnarray*}
with $x\ge 0$, $i=1,2,\dots,n$ and $U(a,b,z)$ represents the Kummer function or confluent hypergeometric function, defined by (Abramowitz and Stegun (1970), eq. 43.2.5),
\begin{equation}\label{confluent}
\displaystyle U(a,b,z)=\int_0^\infty e^{-zt}t^{a-1}(1+t)^{b-a-1}dt,
\end{equation}
with $a,z>0$.
The joint moments of (\ref{MultivariateSibuya}) are,
\begin{eqnarray*}
E(X_1^{r_1}\cdots X_n^{r_n})&=&E(G_1^{r_1})\cdots E(G_n^{r_n})E(H_{\beta,\gamma}^{r_1+\cdots+r_n})\\
&=&\prod_{i=1}^n\frac{\Gamma(\alpha_i+r_i)}{\Gamma(\alpha_i)}\frac{\Gamma(\beta+\bar r)\Gamma(\gamma-\bar r)}{\Gamma(\beta)\Gamma(\gamma)},
\end{eqnarray*}
if $\gamma>\tilde r$ and $\tilde r=\sum_{i=1}^nr_i$.

The distribution of the aggregated risk is given in the following Theorem.
\begin{theorem}
Let $(X_1,\dots,X_n)^\top$ be a multivariate Sibuya distribution defined by the stochastic representation (\ref{MultivariateSibuya}). The distribution of the aggregated risk is,
$$
S_n\sim G_{\bar\alpha}H_{\beta,\gamma},
$$
with pdf,
$$
f_{S_n}(x)=\frac{\Gamma(\tilde\alpha+\gamma)\Gamma(\beta+\gamma)}{\Gamma(\tilde\alpha)\Gamma(\beta)\Gamma(\gamma)}x^{\bar\alpha-1}
U(\tilde\alpha+\gamma,\tilde\alpha-\beta+1,x),\;\;x\ge 0
$$
and $f_{S_n}(x)=0$ if $x<0$, being $\tilde\alpha=\sum_{j=1}^n\alpha_j$, and $U(a,b,z)$ is defined in (\ref{confluent}).
\end{theorem}
\begin{proof}
the proof is direct taking into account that the distribution of the sum is,
$$
S_n=\sum_{j=1}^nG_{\alpha_i}H_{\beta,\gamma}=H_{\beta,\gamma}\sum_{j=1}^nG_{\alpha_i}\stackrel{d}{=}H_{\beta,\gamma}G_{\tilde\alpha},
$$
where $\stackrel{d}{=}$ means equally distributed and $\tilde\alpha=\sum_{j=1}^n\alpha_j$.
\end{proof}\\

The moments of $S_n$ are,
$$
E(S_n^r)=\frac{\Gamma(\tilde\alpha+r)\Gamma(\beta+r)\Gamma(\gamma-r)}{\Gamma(\tilde\alpha)\Gamma(\beta)\Gamma(\gamma)},
$$
if $\gamma>r$.

\section{Conclusions and future research}\label{section6}

In this paper, we have obtained analytic expressions for the probability density function and the cumulative distribution function of aggregated risks, where the risks are modeled according to a mixture of exponential distributions. We have studied some specific models, describing the claims with Pareto (Sarabia et al, 2016), Gamma, Weibull distributions and inverse Gaussian mixture of exponentials (Whitmore and Lee, 1991). We have also proposed an extension of the basic model based on mixtures of gamma distributions. This research can be extended in several different ways. The first obvious extension would be to consider other different Archimedean copulas including the Ali-Mikhail-Haq, Frank and Joe families. A second possibility would be to work with other general classes of Laplace transforms. One of these classes is the family proposed by Hougaard (1986), which includes as a limit case the Laplace transform of the gamma distribution and then the Pareto claims and the Clayton copula. Other extensions include the model discussed in Section \ref{section5}. Eventually, model based on mixtures of Pareto distributions can also be considered. In the Appendix, we include some brief comments about the asymptotic behavior of the aggregated risks in this model. All these points will be addressed in future research.

\section*{Acknowledgements}

The authors thanks to the Ministerio de Econom\'ia y Competitividad (projects ECO2016-76203-C2-1-P, JMS, FP and VJ ECO2013-47092 EGD) for partial support of this work. In addition, this work is part of the Research Project APIE 1/2015-17 (JMS, FP, VJ): ``New methods for the empirical analysis of financial markets" of the Santander Financial Institute (SANFI) of UCEIF Foundation resolved by the University of Cantabria and funded with sponsorship from Banco Santander.


\section*{Appendix}

\subsection*{The partial Bell polynomials 1}\label{appendix1}

The partial or incomplete exponential Bell polynomials are a triangular array of polynomials defined by,
\begin{equation*}
\displaystyle B_{n,k}(x_1,\dots,x_{n-k+1})=\sum\frac{n!}{j_1!\cdots j_{n-k+1}!}\left(\frac{x_1}{1!}\right)^{j_1}\cdots \left(\frac{x_{n-k+1}}{(n-k+1)!}\right)^{j_{n-k+1}},
\end{equation*}
where the sum is taken over all sequences $j_1,\dots,j_{n-k+1}$ of non-negative integers such that $\sum_{i=1}^{n-k+1}j_i=k$ and $\sum_{i=1}^{n-k+1}ij_i=n$.

\noindent In \verb"Mathematica" the partial Bell polynomials are

$$\verb"BellY"[n,k,{x_{1},...,x_{n-k+1}}]$$

\noindent Then for $n=2$ we have,
\begin{eqnarray*}
B_{2,1}(x_1,x_2)&=&x_2,\\
B_{2,2}(x_1)&=&x_1^2,
\end{eqnarray*}
\noindent for $n=3$,
\begin{eqnarray*}
B_{3,1}(x_1,x_2,x_3)&=&x_3,\\
B_{3,2}(x_1,x_2)&=&3x_1x_2,\\
B_{3,3}(x_1)&=&x_1^3,
\end{eqnarray*}
\noindent for $n=4$,
\begin{eqnarray*}
B_{4,1}(x_1,x_2,x_3,x_4)&=&x_4,\\
B_{4,2}(x_1,x_2,x_3)&=&3x_2^2+4x_1x_3,\\
B_{4,3}(x_1,x_2)&=&6x_1^2x_2,\\
B_{4,4}(x_1)&=&x_1^4,
\end{eqnarray*}
\noindent for $n=5$,
\begin{eqnarray*}
B_{5,1}(x_1,x_2,x_3,x_4,x_5)&=&x_5,\\
B_{5,2}(x_1,x_2,x_3,x_4)&=&10x_2x_3+5x_1x_4,\\
B_{5,3}(x_1,x_2,x_3)&=&15x_1x_2^2+10x_1^2x_3,\\
B_{5,4}(x_1,x_2)&=&10x_1^3x_2,\\
B_{5,5}(x_1)&=&x_1^5,
\end{eqnarray*}
\noindent and so on.

\subsection*{The Leibniz and Fa\`a di Bruno formulas}

Some of the formulas for the $n$-th derivative of the Laplace transform can be written as the $n$-th derivative of the product and the composition of functions. The following lemma provides the Leibniz formula.
\begin{lemma}
If $f$ and $g$ are $n$-times differentiable functions, the product $fg$ is also $n$-times differentiable and its $n$-th derivative is given by,
\begin{equation}\label{LeibnizFormula}
\displaystyle\frac{d^n}{dx^n}(fg)(x)=\sum_{k=0}^n{n\choose k}f^{(n-k)}(x)g^{(k)}(x).
\end{equation}
\end{lemma}

Other formulas can be written using the Fa\`a di Bruno formula (Krantz and Parks, 2002), which is an identity that generalized the chain rule to higher derivatives. We consider the version of the  Fa\`a di Bruno formula in terms of the Partial Bell polynomials
\begin{lemma}\label{FaadiBruno}
The $n$th derivative of the composition of two functions $f(g(x))$ can be written as,
\begin{equation}\label{FormulaFaadiBruno}
\displaystyle\frac{d^n}{dx^n}f(g(x))=\sum_{k=1}^nf^{(k)}(g(x))B_{n,k}\left(g'(x),g^{''}(x),\dots,g^{(n-k+1)}(x)\right),
\end{equation}
where $B_{n,k}(x_1,\dots,x_{n-k+1})$ are the partial Bell polynomials.
\end{lemma}

\subsection*{Risk measures: TVaR and tail moments}

In this Section we show how to compute TVaR and in general tail moments for a finite mixture of distributions. First, we provides a Lemma for computing the tail moments.

\begin{lemma}\label{LemmaTailMoments}
Let $X$ be a positive random variable with cdf $F_X(\cdot)$ and $E(X^r)<\infty$, with $r\in\mathbb{N}$. If $a>0$ we have,
$$
\displaystyle E(X^r|X>a)=E(X^r)\frac{1-F_X^{(r)}(a)}{1-F_X(a)},
$$
where $F_X^{(r)}(x)=\frac{\int_0^xt^rf_X(t)dt}{E(X^r)}$ is the cdf of the incomplete $r$th moment.
\end{lemma}
\begin{proof}
The proof is direct since,
$$
\displaystyle E(X^r|X>a)=\frac{\int_a^\infty x^rdF_X(x)}{1-F_X(a)}.
$$
\end{proof}

This theorem can be used in some previous aggregated distributions.
\begin{theorem}
Let $X$ be a positive random variable which is a finite mixture of positive random variables, that is, the cdf of $X$, $F_X(\cdot)$ can be written as,
$$
F_X(x)=\pi_1F_{X_1}(x)+\cdots+\pi_nF_{X_n}(x),\;\;x\ge 0,
$$
where $F_{X_i}(x)$ are the cdf of the random variables $X_i$, $i=1,2,\dots,n$ and $\sum_{i=1}^n\pi_i=1$. Then, if $E(X^r)<\infty$, the $r$th upper tail moment of $X$ can be written as,
$$
E(X^r|X>a)=\sum_{k=1}^n\frac{\pi_k\{1-F_{X_k}(a)\}}{1-F_X(a)}E(X_k^r|X_k>a),
$$
with $a\ge 0$.
\end{theorem}
\begin{proof}
The proof is also direct taking into account,
$$
\displaystyle\int_0^\infty x^rdF_X(x)=\sum_{k=1}^n\pi_k\int_0^\infty x^rdF_{X_k}(x),
$$
together with Lemma \ref{LemmaTailMoments}.
\end{proof}

\subsection*{Asymptotic probabilities for mixtures of Pareto distribution}
Undoubtedly, the single parameter Pareto distribution is commonly used as a basis
for excess of loss quotations as it gives a pretty good description of the
random behaviour of large losses. Moreover, approximating the tail probability of $S_n$ when the losses have Pareto and other distributions is a recurrent problem encountered in the actuarial literature (see, for instance Goovaerts et al. (2005)).

Rohener and Winniwarter (1985) (see also Arnold (1983)) obtained the limiting distribution of an infinite sum of  non i.i.d. classical Pareto variables and linked their results with the theory of stable distributions. Consider the classical Pareto distribution (Arnold (1983) and Rohener and Winniwarter (1985)) with shape parameter $\theta_i>0$ and precision parameter $\beta>0$ and assume that $X_1,\dots,X_n$ are independent but non--identical random variables following a classical Pareto distribution. Rohener and Winniwarter (1985) have proven that as $x\to\infty$,
\begin{eqnarray*}
f_{S_n|\Theta}(x|\theta)\sim \frac{\theta \beta^{m\theta}}{x^{\theta+1}},
\end{eqnarray*}
where $\theta=\min\{\theta_1,\dots,\theta_n\}$ and $m$ denotes the index of the smallest $\theta_i$. Let $\Theta$ be a positive random variable with cdf $F_\Theta(\cdot)$ and Laplace transform $L_{\Theta}(\cdot)$. Thus, the asymptotic behaviour of the unconditional distribution of $S_n$ results
\begin{eqnarray}
f_{S_n}(x)\sim -\frac{d}{dx}L_{\Theta}\left[\log\left(\frac{x}{\beta^m}\right)\right],\label{acpp}
\end{eqnarray}
as it is straightforward to check.

Using (\ref{acpp}), it is easy to get the asymptotic behaviour of the mixture of $f_{S_n|\Theta}(x|\theta)$ with the gamma, $\Theta\sim{\cal G}a(\alpha,\lambda)$, and inverse Gaussian, $\Theta\sim {\cal IG}(\lambda,\mu)$, acting as mixing distributions,  which is expressed as

\begin{eqnarray*}
f_{S_n}(x) & \sim & \frac{\alpha\lambda^{\alpha}}{x\left[\lambda+\log x-m\log\beta\right]^{\alpha+1}},\\
f_{S_n}(x) & \sim & \frac{1}{x}\sqrt{\frac{\lambda}{\varphi_{\lambda,\mu,\beta,m}(x)}}\exp\left[\frac{\lambda}{\mu}-\sqrt{\lambda\varphi_{\lambda,\mu,\beta,m}(x)}\right],
\end{eqnarray*}
respectively, where $\varphi_{\lambda,\mu,\beta,m}(x)=\frac{\lambda}{\mu^2}+2\log\left(\frac{x}{\beta^m}\right)$.

\section*{References}

\begin{description}

\item Abramowitz, M., and Stegun, I.A. (1970). Handbook of Mathematical Functions. Dover Publications. Inc., New York.

\item Albrecher, H., Constantinescu, C., Loisel, S. (2011). Explicit ruin formulas for models with dependence among risks. Insurance: Mathematics and Economics, 48, 265-270.

\item Albrecher, H., Kortschak, D. (2009). On ruin probability and aggregate claim representations for Pareto claim size distributions. Insurance: Mathematics and Economics, 45, 362-373.

\item Arbenz, P., Hummel, C., Mainik, G. (2012).  Copula based hierarchical risk aggregation through sample reordering. Insurance: Mathematics and Economics, 51, 122-133.

\item Arnold, B.C. (1983). Pareto Distributions. International Cooperative Publishing House, Fairland, MD.

\item Arnold, B.C. (2015). Pareto Distributions, Second Edition. Chapman \& Hall/CRC Monographs on Statistics \& Applied Probability, Boca Rat\'on, FL.

\item Asgharzadeh, A., Nadarajah, S. and Sharafi, F. (2017). Generalized inverse Lindley distribution with application to Danish fire insurance data. Communications in Statistics-Theory and Methods, 46, 10 (in press)

\item Asmussen, S., Albrecher, H. (2010). Ruin Probabilities, second ed. World Scientific, New Jersey.

\item Barlow, R.E., Proschan, F. (1981). Statistical Theory of Reliability and Life Testing. Holt, Rinehart and Winston, New York.

\item B{\o}lviken, E., Guillen, M. (2017). Risk aggregation in Solvency II through recursive log-normals. Insurance: Mathematics and Economics, 73, 20-26

\item Coqueret, G. (2014). Second order risk aggregation with the Bernstein copula. Insurance: Mathematics and Economics, 58, 150-158.

\item Cossette, H., C\^ot\'e, M.P., Marceau, E., Moutanabbir, K. (2013). Multivariate distribution defined with Farlie-Gumbel-Morgenstern copula and mixed Erlang marginals: Aggregation and capital allocation. Insurance: Mathematics and Economics, 52, 560-572.

\item C\^ot\'e, M.P., Genest, C. (2015). A copula-based risk aggregation model. The Canadian Journal of Statistics, 43, 60-81.

\item Dacarogna, M., Elbahtouri, L., Kratz, M. (2015). Explicit Diversification Benefit for Dependent Risks. Working paper, ESSEC Business School.

\item Esary, J.D., Proschan, F., Walkup, D.W. (1967). Association of random variables, with applications. Ann. Math. Statist. 38, 1466-1474.

\item Feller, W. (1971). An Introduction to Probability Theory and its Applications, Volume 1. Second Edition. John Wiley, New York.

\item Genest, C., and MacKay, J. (1986). The joy of copulas: bivariate distributions with uniform marginals. The American Statistician, 40, 280-283.

\item Gijbels, I., Herrmann, K. (2014). On the distribution of sums of random variables with copula-induced dependence. Insurance: Mathematics and Economics, 59, 27-44.

\item Gleser, L.J. (1989). The gamma distribution as a mixture of exponential distributions. The American Statistician, 43, 115-117.

\item G\'omez--D\'eniz, E., Sarabia, J.M. and Balakrishnan, N. (2012). A multivariate discrete Poisson-Lindley distribution: extensions and actuarial applications. ASTIN Bulletin, 42, 2, 655-678.

\item Good, I.J. (1953). The population frequencies of species and the estimation of population parameters. Biometrika 40, 237-260.

\item Goovaerts, M., Kaas, R., Laeven, R., Tang, Q. and Vernic, R. (2005). The tail probability of discounted sums of Pareto-like losses in insurance. Scandinavian Actuarial Journal, 6, 446-461.

\item Gradshteyn, L.S., Ryzhik, I. M. (1980). Table of Integrals, Series and Products, New York: Academic Press.

\item Guill\'en, M., Sarabia, J.M., Prieto, F. (2013). Simple risk measure calculations for sums of positive random variables. Insurance: Mathematics and Economics, 53, 273-280.

\item Jewell, N.P. (1982). Mixtures of exponential distributions. The Annals of Statistics, 10, 479-484.

\item Hashorva, E., Ratovomirija, G. (2015). On Sarmanov Mixed Erlang Risks in Insurance Applications. ASTIN Bulletin, 45, 175-205.

\item Hougaard, P. (1986). Survival models for heterogeneous populations derived from stable distributions. Biometrika, 73, 387-396.

\item Johnson, N.L., Kotz, S., Balakrishnan, N. (1994). Continuous Univariate Distributions. Vol.1, second ed., John Wiley, New York.

\item Klugman, S.A., Panjer, H.H., Willmot, G.E. (2008). Loss Models. From Data to Decisions. Third Edition. John Wiley, New York.

\item Krantz, Steven G., Parks, Harold R. (2002). A Primer of Real Analytic Functions, Birkh\"auser Advanced Texts - Basler Lehrb\"ucher (Second ed.), Boston: Birkh\"auser Verlag,

\item Lee, M-L.T., Gross, A.J. (1989). Properties of conditionally independent generalized gamma distributions. Probability in the Engineering and Informational Sciences, 3, 289-297.

\item  Lindley, D.V., (1958). Fiducial distributions and Bayes's theorem. Journal of the Royal Statistical Society. Series B, 20, 102-107.

\item Lindley, D.V., Singpurwalla, N.D. (1986). Multivariate distributions for the life lengths of components of a system sharing a common environment. Journal of Applied Probability, 23, 418-431.

\item McDonald, J.B. (1984). Some generalized functions for the size distribution of income. Econometrica, 52, 647-663.

\item Nayak, T.K. (1987). Multivariate Lomax distribution: properties and usefulness in reliability theory. Journal of Applied Probability, 24, 170-177.

\item Nelsen,R. (1999). An Introduction to Copulas. Springer, New York.

\item Oakes, D. (1989). Bivariate survival models induced by frailties. Journal of the American Statistical Association, 84,487-493.

\item Roehner, B. and Winniwarter, P. (1985). Aggregation of independent Paretian random variables. Adv. Appl. Prob., 17, 465-469.

\item Roy, D., Mukherjee, S.P. (1988). Generalized mixtures of exponential distributions. Journal of Applied Probability, 25, 510-518.

\item Sarabia, J.M., G\'omez-D\'eniz, E., Prieto, F., Jord\'a, V. (2016). Risk aggregation in multivariate dependent Pareto distributions. Insurance: Mathematics and Economics, 71, 154-163.

\item Sibuya, M. (1979). Generalized Hypergeometric, Digamma and Trigamma Distributions. Annals of the Institute of Statistical Mathematics, 31, 373-390.

\item Vernic, R. (2016). On the distribution of a sum of Sarmanov distributed random variables. Journal of Theoretical Probability, 29, 118-142.

\item Watson, G.N. (1995). A Treatise on the Theory of Bessel Functions, 2nd ed. Cambridge University Press.

\item Whitmore, G.A. (1988). Inverse Gaussian mixtures of exponential distributions, unpublished paper, McGill University, Faculty of management.

\item Whitmore, G.A., Lee, M.-L.T. (1991). A multivariate survival distribution generated by an inverse Gaussian mixture of exponentials. Technometrics, 33, 39-50.

\end{description}

\end{document}